\documentclass{article}
\usepackage{spconf,amsmath,graphicx}
\usepackage{amssymb}
\usepackage{amsbsy}\usepackage{epsfig}
\usepackage{cite}
\usepackage{balance}
\usepackage{color}
\usepackage{float}
\floatstyle{ruled}
\newfloat{algorithm}{tbp}{loa}
\providecommand{\algorithmname}{Algorithm}
\usepackage{algorithm}
\usepackage[noend]{algpseudocode}
\newcommand{\algalign}[2]

\newcommand{\beq}{\begin{equation}}
\newcommand{\eeq}{\end{equation}}

\def\bw{\mbox{\boldmath $w$}}

\def\bx{\mathbf{x}}

\def\bX{\mbox{\boldmath $X$}}

\def\mx{\mbox{\boldmath $S$}}

\def\bX{\mathbf{X}}

\def\b1{\mbox{\boldmath $1$}}
\def\b0{\mbox{\boldmath $0$}}

\def\mx{\mbox{$\mathbf{x}$}}
\def\mb{\mbox{$\mathbf{b}$}}

\def\mB{\mbox{$\mathbf{B}$}}

\def\mG{\mbox{$\mathbf{G}$}}

\def\mL{\mbox{$\mathbf{L}$}}

\newcommand{\ds}{\displaystyle}

\newtheorem{theorem}{\textbf{Theorem}}

\newenvironment{proof}[1][Proof]{\noindent \textbf{#1.} }{\qedsymbol}
\newcommand{\qedsymbol}{\hspace{\fill}\rule{1.5ex}{1.5ex}}
\floatstyle{ruled}
\newfloat{algorithm}{tbp}{loa}
\providecommand{\algorithmname}{Algorithm}
\floatname{algorithm}{\protect\algorithmname}

\DeclareMathOperator*{\argmin}{arg\,min}
\title{TOPOLOGICAL SIGNAL PROCESSING OVER WEIGHTED SIMPLICIAL COMPLEXES}
\name{Claudio Battiloro, Stefania Sardellitti, Sergio Barbarossa, Paolo Di Lorenzo\thanks{This work was supported in part by H2020 EU/Taiwan Project 5G CONNI  Nr. AMD-861459-3 and in part by MIUR under the PRIN Liquid-Edge contract.}}
\address{DIET Department, Sapienza University of Rome, Via Eudossiana 18, 00184, Rome, Italy \\
E-mail:  \{claudio.battiloro, stefania.sardellitti, sergio.barbarossa,  paolo.dilorenzo\}@uniroma1.it 
}
\begin{document}
\ninept
\maketitle
\vspace{-.2cm}
\begin{abstract}
Weighing the topological domain over which data can be represented and analysed is a key strategy in many signal processing and machine learning applications, enabling the extraction and exploitation of meaningful data features and their (higher order) relationships. Our goal in this paper is to present topological signal processing tools for weighted simplicial complexes. Specifically, relying on the weighted Hodge Laplacian theory, we propose efficient strategies to jointly learn the weights of the complex and the filters for the solenoidal, irrotational and harmonic components of the signals defined over the complex. We numerically asses the effectiveness of the proposed procedures.

\end{abstract}
\vspace{-.2cm}
\section{Introduction}
\label{sec:intro}
\vspace{-.1cm}
In the last years, there has been a growing interest in the processing of signals defined over topological spaces  \cite{carlsson2009topology}, \cite{munkres2000topology}, i.e. over domains composed of a set of points along with a set of neighborhood relations among them, not necessarily metric. A renowned example are graph signals, usually processed with tools from the Graph Signal Processing (GSP) framework \cite{shuman2013,ortega2018graph}. However, graphs encode only pairwise relationships between data; on the contrary, many real-world phenomena involve multi-way relationships as, e.g., in biological or social networks. Recently, the  Topological Signal Processing (TSP)  over Simplicial Complexes framework has been proposed in \cite{barb_2020,barb_Mag_2020}, and it represents a proper generalization of the GSP  framework for the representation and analysis of signals defined over simplicial complexes; in \cite{SCHAUB2021}, the authors presented  a tutorial on the emerging field of signal processing over hypergraphs and simplicial complexes. A simplicial complex is a topological space composed by a set of elements $\mathcal{V}$ and a set  $\mathcal{S}$  containing subsets of various cardinality of the elements of $\mathcal{V}$ satisfying the inclusion property; the rich algebraic structure of simplicial complexes make them  particularly suited to capture multiway relations among data. Simplicial-based processing methods have been applied in many fields, such as statistical ranking \cite{jiang2011},  tumor progression analysis \cite{roman2015simplicial},  and brain \cite{andjelkovic2020topology} and  biological  \cite{estrada2018centralities} networks. For this reason, there was also a raising interest in the development of (deep) neural network architectures able to handle data defined on simplicial complexes \cite{san,bodnar2021weisfeiler,scn}. Recently, weighted simplicial complexes (WSC) have been investigated as powerful tools to capture the information on data by assigning weights to each simplex in the complex; in this work, we focus on WSCs to develop signal processing procedures for simplex-structured data.

\noindent\textbf{Related Works.} In \cite{bianconi2022weighted}, the authors showed that weighted simplicial complexes provided  with a proper choice of weights can be used to capture higher-order relationships in network data. A non-equilibrium model for weighted growing simplicial complexes is developed in  \cite{bianconi2017weighted}, while in \cite{sharma2017weighted} the authors focused on the problem of group recurrence prediction eploiting  WSCs. Weighted Laplacians for weighted simplicial complexes are introduced in \cite{Wu2018} as a generalization of the combinatorial Hodge Laplacians. In \cite{meng2020weighted}, weighted persistent homology is proposed to analyse  biomolecular data where weights reflect certain physical, chemical or biological properties into the simplicial complex generation. The works in \cite{san,hajiji2022att} proposed attentional deep architectures for simplex-structured data that can be also seen as methods implicitly working on WSCs. 

\noindent\textbf{Contribution.} Our goal in this paper is to  establish the fundamental tools for processing signals over weighted simplicial complexes. Weighing the topological domain over which data are represented and processed enables the extraction of meaningful data features and data relationships encoded by  the weights  assigned  to  simplices of different order. We will show that weighing the simplices corresponds to introducing a metric on the simplicial complex; for this reason, a weighted Hodge Laplacian is introduced by taking into account the resulting  metric tensors. To enhance the advantages of working on WSCs, we propose an edge flow estimation strategy to jointly learn the metric tensor and the edge flow from observed noisy data. Moreover, we propose an efficient strategy to learn the metric from data by minimizing the observed signal total variations, i.e. the circulation of the signals along the triangles. We believe that several other techniques could be developed to give both theoretical and practical contributions.

\vspace{-.2cm}
\section{WEIGHTED SIMPLICIAL COMPLEXES}
\label{sec:wsc}
\vspace{-.1cm}
In this section we  introduce the algebraic representation of  \textit{weighted} simplicial complexes  and the fundamental  tools to analyze signals defined over these topological spaces.\\
\noindent\textbf{simplicial Complexes.} Given a finite set $\mathcal{V}= \{ v_{i}\}_{i=0}^{N-1}$ of $N$ vertices, a {\it $k$-simplex} $\sigma^k_i$ is formed by an unordered set of  $k+1$ vertices in $\mathcal{V}$.
A \textit{face} of the $k$-simplex
is a $(k-1)$-simplex and every $k$-simplex has exactly $k+1$ face. 
An \textit{abstract simplicial complex} ${\cal X}$ is definited as a finite collection of simplices  closed under inclusion of the faces of each simplex, i.e., if ${\cal \sigma}_i \in {\cal X}$, then all faces of $\sigma_i$ also belong to ${\cal X}$.
The order of a simplex is one less than its cardinality and  the order of a simplicial complex is the order of its highest order simplex.
An abstract simplicial complex can be embedded into a metric space. If the complex is embedded into the Euclidean space, a vertex is a 0-dimensional simplex, a line segment  has dimension $1$, a triangle is a simplex of order $2$ and so on. A graph is simply a simplicial complex of order one.
The structure of a simplicial complex is captured by the neighborhood relations of its subsets. Specifically, two simplices of order $k$, $\sigma_i^k, \sigma_j^k \in {\cal X}$, are {\it upper adjacent} in ${\cal X}$, if they are both faces of a simplex of order $k+1$, while they are {\it lower adjacent} in ${\cal X}$, if they share a common face of order $k-1$. We usually focus on  second order simplicial complexes, denoted with $\mathcal{X}=\{\mathcal{V},\mathcal{E},\mathcal{T}\}$ where $\mathcal{V}$, $\mathcal{E}$, $\mathcal{T}$ are the sets of  $0$, $1$ and $2$-simplices, i.e. vertices, edges and triangles, respectively. 

\noindent\textbf{Incidence Matrices.} Let us  denote by $n_k$ the number of simplices of order $k$ in  the complex. 
Given an orientation of all simplices (see \cite{Barbarossa2020TopologicalSP} for details), the structure of a simplicial complex ${\cal X}$ of dimension $K$ is captured by the set of its incidence (or boundary)  matrices $\mathbf{B}_k \in \mathbb{R}^{n_{k-1}\times n_k}$, $k=1, \ldots, K$,  with entries $B_k(i,j)=0$ if $\sigma^{k-1}_i$ is not a face of $\sigma^{k}_j$, and $B_k(i,j)=1$  (or $-1$), if $\sigma^{k-1}_i$ is a face of $\sigma^{k}_j$ and its orientation is  coherent (or not) with the orientation of $\sigma^{k}_j$. We denote the  set of $k$-simplex in $\mathcal{X}$ as  ${\cal D}_{k} := \{\sigma_i^{k}: \sigma_i^{k} \in \mathcal{X}\} $, with $|{\cal D}_{k}| = n_k$ and, obviously, ${\cal D}_{k} \subset {\cal X}_{K}$.

\noindent\textbf{Simplicial Signals.} We are interested in processing signals defined over a simplicial complex. A $k$-simplicial signal $\mathbf{x}_{k}$ is usually defined as a collection of mappings from the set of all $k$-simplices contained in the complex to real numbers:
\begin{equation}\label{signals}
    \mathbf{x}_k = [x_k(\sigma_1^k),\dots,x_k(\sigma_i^k), \dots, x_k(\sigma_{n_k}^k)] \in \mathbb{R}^{n_k},
\end{equation}
where $x_{k}: {\cal D}_{k} \rightarrow \mathbb{R}$. Although the definition in \eqref{signals} is formally correct, it can be reformulated using tools from algebraic topology, that we explicitly need for defining WSCs.  In particular, we need the notions of chains, cochains and metric tensors.

\noindent\textbf{Weighted Simplicial Complexes.} A $k$\textit{-chain} $\tau^k$ is a linear combination of $k$-simplices \cite{grady2010}:
\begin{equation}\label{kchains}
\tau^k=\sum_{j=1}^{n_k}  c_{j} \sigma_j^k, c_j \in \mathbb{R}.
\end{equation}
The space of all the $k$-chains, denoted with $\mathcal{C}_k$, is a real vector space with a basis given by the collection of $k$-simplices, which we refer to  as \textit{basic} $k$-chains.  Being a finite dimensional vector space, the chain space $\mathcal{C}_k$ can be equipped with an inner product completely determined by its basic $k$-chains:
\beq \label{inn_prod_basis}
g_{ij}^{k} := \langle \sigma_i^k, \sigma_j^k\rangle_{\mathcal{C}_k}, 
\eeq
where $g_{ij}^{k} \in \mathbb{R}$, $g_{ij}^{k}=g_{ji}^{k}$, $i,j=1,\ldots, n_k$. In this way, given two chains $\tau^k$ and $\gamma^k$ with coefficients $\{c_j\}_j$ and $\{a_j\}_j$, respectively, we obtain:
\beq \label{inn_prod_chains}
\langle \tau^k, \gamma^k\rangle_{\mathcal{C}_k} = \sum_i\sum_j c_i a_jg_{ij}^{k}. 
\eeq
In this work, we assume orthogonality and positiveness, meaning $g_{ij}^{k}=0$, $\forall\,  i \neq j$, and $g_{ii}^{k}>0$ $\forall\,  i$. We refer to the set of all inner products in \eqref{inn_prod_basis}, for $i, j = 1, \ldots n_k$, as the \textit{metric tensor} of order $k$. The dual space $\mathcal{C}^k$ of $\mathcal{C}_k$ is the space of all linear functional $\tau^{* k}$ from $\mathcal{C}_k$ to $\mathbb{R}$: we refer to these linear functionals as $k$\textit{-cochains}. Due to to the canonical isomorphism, the metric tensor induces an inner product also on the dual space; in particular, given two  cochains $\tau^{* k}$ and $\gamma^{* k}$ with coefficients $\{c^j\}_j$ and $\{a^j\}_j$, respectively, we have:
 \beq \label{inn_prod_cochains}
\langle \tau^{*k}, \gamma^{*k}\rangle_{\mathcal{C}^k} = \sum_i\sum_j c^ia^jw_{ij}^{k},
\eeq
 where $w_{ij}^k=1/g_{ij}^k$. It can be proven that the cochain space $\mathcal{C}^k$ is naturally isomorphic to $\mathbb{R}^{n_k}$ \cite{grady2010}, so that we can identify a cochain $\tau^{*,k}$ with a vector $\mathbf{x}_k = [c^1,\dots,c^{n_k}] \in \mathbb{R}^{n_k}$ containing the coefficients of its corresponding chain. At this point, it is sufficient to set $\mathbf{x}_k(i) = x_k(\sigma_i^k) = c^i$ for re-obtaining the definition in \eqref{signals}; therefore, we can state that simplicial signals and co-chains are the same object (up to an isomorphism). As a direct  consequence of the aforementioned results, we can see the metric tensor  as a  diagonal matrix $\mathbf{G}_k$ with positive entries given by $g_{i i}^{k}$ for $ 1 \leq i\leq n_k$. 
Therefore, given two signals $\mx^{k}_1,\mx^{k}_2$ defined over $k$-simplices, their inner product   is defined as:
\beq \label{eq:scal_prod}
\langle \mx^{k}_1,\mx^{k}_2 \rangle= \mx^{k \, T}_1 \mG^{-1}_k \mx^{k}_2=\sum_{i=1}^{n_k} w_{i i}^{k} \mx^{k}_1(i)\mx^{k}_2(i).
\eeq
We define a \textit{weighted simplicial complex} as a simplicial complex whose chain spaces are equipped with non-trivial metric tensors (non-identity matrices). 

\noindent\textbf{Hodge decomposition.}
To find an algebraic representation of the weighted simplicial complex that is able to capture its  topological and metric structures, we first need  to introduce the  boundaries and coboundaries operators. 
The $k-$boundary operator $\boldsymbol{\partial}_{k}: \mathcal{C}_{k} \rightarrow \mathcal{C}_{k-1}$ is a linear operator  mapping $k$-chains to $(k-1)$-chains, and we denote its dual with $\boldsymbol{\delta}_{k}:=\boldsymbol{\partial}_{k}^*$. The  dual $\boldsymbol{\delta}_{k}$ is called the 
$k-$coboundary operator, and it  maps $(k-1)$-cochains to $k$-cochains. It can be proven that \cite{grady2010}:
\begin{equation}\label{cobound}
\boldsymbol{\delta}_{k}=\mB_k^T.
\end{equation}
We can derive an expression for the adjoint (of the dual) operator $\boldsymbol{\delta}_{k}^{\prime}$ as a function of the metric tensor, observing that it holds that:
\begin{equation}\label{dual_inner}
\langle \mx^{k}, \boldsymbol{\delta}_{k}\mx^{k-1} \rangle=\langle  \boldsymbol{\delta}_{k}^{\prime}\mx^{k}, \mx^{k-1}\rangle
\end{equation}
for every pairs of signals $\mx^{k-1} \in \mathbb{R}^{n_{k-1}}$, $\mx^{k} \in \mathbb{R}^{n_{k}}$. Then,  combining \eqref{cobound} and \eqref{dual_inner}, we can easily write  $\boldsymbol{\delta}_k^{\prime}$ as:
\beq \label{adj_coboundary}
\boldsymbol{\delta}_{k}^{\prime} = \mG_{k-1} \mB_k \mG_{k}^{-1}.\eeq
The topological structure of a (weighted or not) $K$-simplicial complexes is fully described by the higher order Hodge Laplacian matrices of order $k=1, \ldots, K$, defined as:
\beq \label{hodge_top}
\mL_k=\boldsymbol{\delta}_{k}\boldsymbol{\delta}_{k}^{\prime}+\boldsymbol{\delta}_{k+1}^{\prime} \boldsymbol{\delta}_{k+1}. 
\eeq
Specifically, using the expression of the adjoint coboundaries in (\ref{adj_coboundary}), we easily get:
\beq \label{hodge_laps}
\begin{array}{ll}
\mL_0=\mG_0  \mB_1 \mG_{1}^{-1} \mB_1^T \medskip\\
\mL_k= \mB_k^T \mG_{k-1} \mB_k \mG_k^{-1}+ \mG_{k}\mB_{k+1} \mG_{k+1}^{-1} \mB_{k+1}^T, \\
\mL_K= \mB_K^T \mG_{K-1} \mB_K \mG_{K}^{-1},
\end{array}
\eeq
$k\!=\!1,\!\ldots,\!K-1$. Then, for instance, the first-order Laplacian for a simplicial complex of order $2$ can be written as:
\beq
\mL_1= \mB_1^T \mG_{0} \mB_1 \mG_1^{-1}+ \mG_1 \mB_{2} \mG_{2}^{-1} \mB_{2}^T.
\eeq
Note that defining the  lower and upper Laplacians as $\mL_{k,d}=\mB_k^T \mG_{k-1} \mB_k \mG_k^{-1}$
and $\mL_{k,u}=\mG_{k}\mB_{k+1} \mG_{k+1}^{-1} \mB_{k+1}^T$, respectively, it holds that $\mL_{k,d} \mL_{k,u}= \mL_{k,u} \mL_{k,d}=\mathbf{0}$. 
This fact implies that an Hodge  decomposition  \cite{frankel2011geometry, desbrun2006discrete, Bell2008} holds for the signals space $\mathbb{R}^{n_k}$; therefore, it can be decomposed as:
\beq\label{hodge_dec}
\mathbb{R}^{n_k}=\text{im}(\boldsymbol{\delta}_k) \oplus \text{ker}(\mL_k) \oplus \text{im}(\boldsymbol{\delta}_{k+1}^{\prime}).
\eeq
Then, exploiting \eqref{adj_coboundary} and \eqref{hodge_dec}, a higher order signal $\mx^{k} \in \mathbb{R}^{n_k}$ can be decomposed in the sum of three orthogonal components
\beq \label{eq:Hodge_dec}
\mx^{k}= \mB_k^{T}\mx^{k-1}+\mG_{k} \mB_{k+1} \mG_{k+1}^{-1} \mx^{k+1}+ \mx_{\textrm{h}}^k
\eeq
where, for $k=1$, $\mx^{1}_{\textrm{irr}}=\mB_1^{T}\mx^{0}$, $\mx^{1}_{\textrm{sol}}=\mG_{1} \mB_2 \mG_2^{-1} \mx^{2}$
and, $\mx_{\textrm{h}}^1 \in \text{ker}(\mL_1)$ are  the irrotational, solenoidal and harmonic flows, respectively \cite{barb_2020}.  In the sequel, we will focus, w.l.o.g., on edge flow signals and complexes of order 2 composed of $N$ nodes, $E$ edges and $T$ triangles. Therefore, we will drop the subscripts and denote $\mathbf{x}^1$ with $\mathbf{x}$, $\mathbf L_1$ with $\mathbf L$, $\mathbf L_{1,d}$ with $\mathbf L_d$, and $\mathbf{x}^1_h$ with $\mathbf{x}_h$.

\section{Joint learning of edge flows and Weights}
Let us suppose to observe an edge flow signal affected by AWG noise, defined as $\widetilde{\mathbf{x}} = \mathbf{x} + \mathbf{n}$, where $\mathbf{x}$ denotes the clean flow, whereas $\mathbf{n}\overset{i.i.d.}{\sim} \mathcal{N}(0,\sigma^2)$ denotes the noisy flow. In this section, we formulate a denoising problem as a constrained problem, rooted in
the Hodge decomposition, and we propose an efficient strategy for jointly learning the weights (metric tensor) $\mathbf{G}_2$ associated with the $2$-order simplices (triangles) and the flow $\mathbf{x}$. In this first study, we  consider the weighing of the nodes $\mathbf{G}_0$ and of the edges $\mathbf{G}_1$ as given.  
Based on the decomposition in \eqref{eq:Hodge_dec}, we can model the observed flow as:
\begin{align}
\label{flow_model}
    & \mathbf{x} = \mathbf{B}_1^T\mathbf{x}^0+ \mathbf{G}_1\mathbf{B}_2\mathbf{G}_2^{-1}\mathbf{x}^2 + \mathbf{x}_{\textrm{h}}, \nonumber \\ 
    & \widetilde{\mathbf{x}} = \mathbf{x}+ \mathbf{n}.
\end{align}
We formulate the denoising problem as follows:
\begin{align}
\mathcal{Q}) \; \, &(\widehat{\mathbf{x}}^{0},\widehat{\mathbf{x}}^{2},\widehat{\mathbf{x}}_\textrm{h}, \widehat{\mathbf{G}}_2) = \nonumber\\
&\underset{\mathbf{x}^{0},\mathbf{x}^{2},\mathbf{x}_{\textrm{h}},\mathbf{G}_2}{\argmin} \|\mathbf{B}_1^T\mathbf{x}^0 + \mathbf{G}_1\mathbf{B}_2\mathbf{G}_2^{-1}\mathbf{x}^2 + \mathbf{x}_{\textrm{h}} - \widetilde{\mathbf{x}}\|^{2} \nonumber \\
& \qquad \textrm{s.t.} \quad\;\;\, \text{a)} \, \mathbf{L}\mathbf{x}_{\textrm{h}} = \mathbf{0}, \nonumber  \\
& \qquad \qquad  \; \;\;  \text{b)} \,[\mathbf{G}_2^{-1}]_{ii} > 0, \; [\mathbf{G}_2^{-1}]_{ij} = 0, \forall \, i \neq j 
\end{align}
where  the constraint  $\text{a)}$ 
 forces $\widehat{\mathbf{x}}_\textrm{h}$  to belong to the kernel of the Laplacian (the harmonic subspace), while the constraints  $\text{b)}$ impose the diagonal structure to  $\widehat{\mathbf{G}}_2$  with positive entries.  Problem $\mathcal{Q}$ is  not jointly convex, 
 but it is block multi-convex, i.e. convex with respect to each of the optimization variables while holding all others fixed. For this reason, we propose an efficient iterative alternating minimization algorithm to find local optimal solutions. Denoting  with $t$ the iteration index, we initialize our iterative algorithm  with a random feasible point  $(\widehat{\mathbf{x}}^0[t],\widehat{\mathbf{x}}^2[t],\widehat{\mathbf{x}}_{\textrm{h}}[t],\widehat{\mathbf{G}}_2[t])$ at time $t=0$.  Then, defining the point $\hat{\mathbf{z}}[t]:=(\widehat{\mathbf{x}}^0[t],\widehat{\mathbf{x}}^2[t],\widehat{\mathbf{x}}_{\textrm{h}}[t])$, the proposed alternating optimization method consists in solving at each iteration $t$ the two following convex problems:
\begin{align}\label{sig_prob}
 \mathcal{Q}_1) \;\hat{\mathbf{z}}[t] = 
& \!\!\!\underset{\mathbf{z}=(\mathbf{x}^0,\mathbf{x}^2,\mathbf{x}_{\textrm{h}})}{\argmin} \!\|\mathbf{B}_1^T\mathbf{x}^0 + \mathbf{G}_1\mathbf{B}_2\widehat{\mathbf{G}}_2^{-1}[t-1]\mathbf{x}^2 + \mathbf{x}_{\textrm{h}} - \widetilde{\mathbf{x}}\|^{2} \nonumber \\
& \qquad \textrm{s.t.}  \quad\;\;\, \mathbf{L}[t-1]\mathbf{x}_{\textrm{h}} = \mathbf{0},
\end{align}
where $\mathbf{L}[t-1]=\mathbf{L}_d+\mathbf{G}_1\mathbf{B}_2 \widehat{\mathbf{G}}_2^{-1}[t-1]\mathbf{B}_2^T$, and
\begin{align}\label{metric_prob}
\!\!\! \mathcal{Q}_2) \;\widehat{\mathbf{G}}_2[t] =  
& \underset{\mathbf{G}_2}{\argmin} \|\mathbf{B}_1^T\widehat{\mathbf{x}}^0[t]+ \mathbf{G}_1\mathbf{B}_2\mathbf{G}_2^{-1}\widehat{\mathbf{x}}^2[t] + \widehat{\mathbf{x}}_\textrm{h}[t] - \widetilde{\mathbf{x}}\|^{2} \nonumber \\
& \qquad \textrm{s.t.} \; \;(\mathbf{L}_d+ \mathbf{G}_1\mathbf{B}_2\mathbf{G}_2^{-1}\mathbf{B}_2^T)\widehat{\mathbf{x}}_{\textrm{h}}[t] = \mathbf{0}, \; \nonumber \\& \qquad \qquad [\mathbf{G}_2^{-1}]_{ii} > 0, \; [\mathbf{G}_2^{-1}]_{ij} = 0, \forall i\neq j.
\end{align}
 Problems $\mathcal{Q}_1$ and $\mathcal{Q}_2$ are convex, and  can be efficiently solved with any numerical solver. Furthermore, using similar derivations as in  \cite{barb_2020}, we can easily prove that problem $\mathcal{Q}_1$ admits the following closed form solution:
\begin{align}
    & \widehat{\mathbf{x}}^0[t] = \mathbf{L}_0^{\dag}\mathbf{B}_1\widetilde{\mathbf{x}}, \forall \, t \label{0_sol}\\
    & \widehat{\mathbf{x}}^2[t] =  (\widehat{\mathbf{G}}_2^{-1}[t-1]\mathbf{B}_2^T \mathbf{G}_1 \mathbf{B}_2\widehat{\mathbf{G}}_2^{-1}[t-1])^{\dag}\widehat{\mathbf{G}}_2[t-1]^{-1}\mathbf{B}_2^T\widetilde{\mathbf{x}}, \label{1_sol}  \\
    &\widehat{\mathbf{x}}_\textrm{h}[t] = \widetilde{\mathbf{x}} - \mathbf{B}_1^T\widehat{\mathbf{x}}^0[t] - \mathbf{G}_1\mathbf{B}_2\mathbf{G}_2^{-1}[t-1]\widehat{\mathbf{x}}^2[t],\label{h_sol}
\end{align}
where $\mathbf{L}_0^{\dag}$ is the Moore-Penrose pseudoinverse.
Problem $\mathcal{Q}$ can be also regularized (with convex penalties), leading to the same procedure but with the additional regularization terms in the objective functions. The proposed procedure is listed in Algorithm 1.\\
\begin{algorithm}[t]
   \caption{: EDGE FLOW ESTIMATION }    \hspace*{\algorithmicindent} \textbf{Inputs}: \\
    \hspace*{\algorithmicindent} \quad $\widetilde{\mathbf{x}} \in \mathbb{R}^{E}$: Noisy edge flow signal. \\
    \hspace*{\algorithmicindent} \quad $\mathbf{B}_1 \in \mathbb{R}^{N\times E}$: Nodes to edges incidence matrix \\
    \hspace*{\algorithmicindent} \quad $\mathbf{B}_2 \in \mathbb{R}^{E\times T}$: Edges to triangles incidence matrix \\
    \hspace*{\algorithmicindent} \quad $\mathbf{G}_1 \in \mathbb{R}^{E\times E}$: Edges weights (metric tensor) \\
    \hspace*{\algorithmicindent} \quad $\widehat{\mathbf{x}}^0[0]$, $\widehat{\mathbf{x}}^2[0]$, $\widehat{\mathbf{x}}_{\textrm{h}}[0]$, $\widehat{\mathbf{G}}_2[0]$: Estimates initializations \\
    \hspace*{\algorithmicindent} \quad $N_{t}$: number of iterations (can be replaced by stopping criterion) \\
    \hspace*{\algorithmicindent} \textbf{Outputs}: \\
    \hspace*{\algorithmicindent} \quad $\widehat{\mathbf{x}}^0$, $\widehat{\mathbf{x}}^1$, $\widehat{\mathbf{x}}_{\textrm{h}}$, $\widehat{\mathbf{G}}_{2}$: Learned signals and weights (metric tensor)
    \begin{algorithmic}[1]
        \Function{Edge flow estimation \,} {\textbf{Inputs}}
                \For{$t \in [1, N_t]$}
                    \State $\widehat{\mathbf{x}}^0[t]$, $\widehat{\mathbf{x}}^2[t]$, $\widehat{\mathbf{x}}_{\textrm{h}}[t]$: Compute \eqref{0_sol}, \eqref{1_sol}, and \eqref{h_sol}
                    \State $\widehat{\mathbf{G}}_2[t]$: Numerically solve $\mathcal{Q}_2$
                \EndFor
        \Return: \\
        \hspace*{\algorithmicindent} \quad $\widehat{\mathbf{x}}^0=\widehat{\mathbf{x}}^0[N_t]$ \\ 
\hspace{\algorithmicindent} \quad $\widehat{\mathbf{x}}^2=\widehat{\mathbf{x}}^2[N_t]$ \\
\hspace{\algorithmicindent} \quad $\widehat{\mathbf{x}}_{\textrm{h}} = \widehat{\mathbf{x}}_{\textrm{h}}[N_t]$ \\
\hspace{\algorithmicindent} \quad $\widehat{\mathbf{G}}_2=\widehat{\mathbf{G}}_2[N_t]$
       \EndFunction
\end{algorithmic}
\end{algorithm}\label{algo:proj}
To numerically test the effectiveness of the proposed edge flow estimation strategy, we consider a random simplicial complex with $N=40$ nodes, $E=137$ edges, $T=96$ triangles, and with the metric tensors $\mathbf{G}_1$ and $\mathbf{G}_2$ being random positive diagonal matrices. We generate  random noisy edge signals $\widetilde{\mathbf{x}}$  according to the model in \eqref{flow_model} with  $\mathbf{x}^0,\mathbf{x}^2,\mathbf{x}_{\textrm{h}}$ being random sparse vectors. Then, we apply the proposed alternating minimization scheme to estimate $\mathbf{x}$ from $\widetilde{\mathbf{x}}$. For this experiment, we also regularize Problem $\mathcal{Q}$ with a $l_1$ penalty on the signal components. In Figure \ref{nmse_flow}, we show the  correlation coefficient $\rho = \frac{|\widehat{\mathbf{x}}^T\mathbf{x}|}{\|\widehat{\mathbf{x}}\|\|\mathbf{x}\|}$ (to neglect the effect of multiplicative constants) versus the noise standard deviation $\sigma$, comparing our method against the estimation only of the signals components assuming a flat metric (unitary weights) $\mathbf{G}_2 = \mathbf{I}$; the results are averaged over $20$ signals and noise realizations. As the reader can notice, the joint learning of the metric tensor and the signals components show a significant performance gain. We plan to extend this work by designing more complex procedures involving the learning also of the weights $\mathbf{G}_0$ and $\mathbf{G}_1$, as well as testing on real data.
\begin{figure}[t]
\centering
\includegraphics[width=0.48\textwidth]{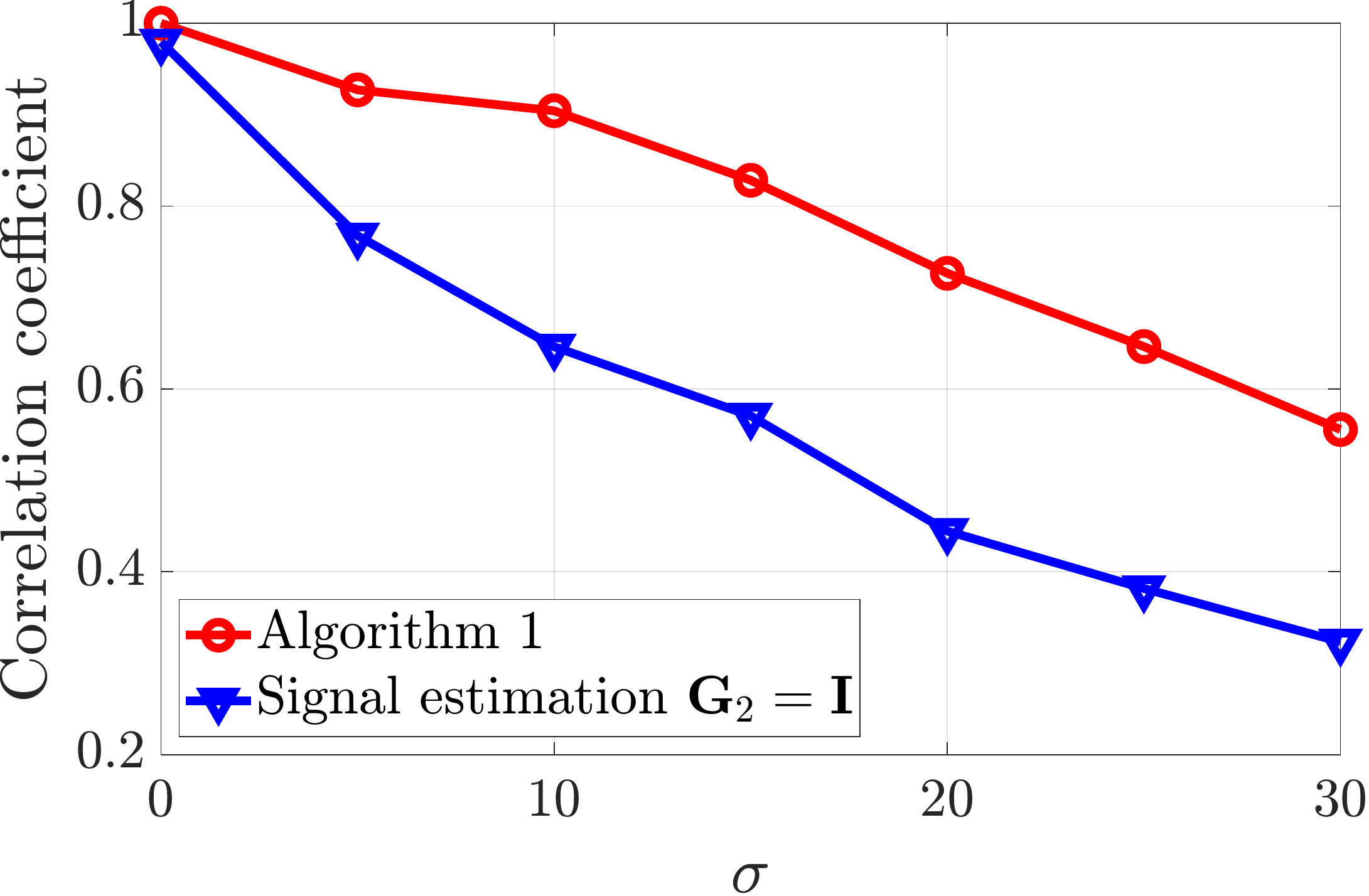}
\caption{Correlation coefficient vs the noise standard deviation.}
\label{nmse_flow}
\end{figure} 
\section{Learning the metric tensor from data}
In this section we propose an efficient strategy to learn the metric tensor $\mathbf{G}_{2}$ from a set of  observed edge signals.  Assuming that the  edge flows are smooth over the solenoidal subspace, so that their circulation along the triangles  is minimum, we formulate the learning of the  metric tensor  as a total variation minimization problem.  We start from the observation of  $M$ snapshots of edge signals  $\bx(m)$, $m=1,\ldots,M$. The squared norm of the circulation of each signal along the triangles of the complex can be written as 
\beq \label{eq:TV_m}
\begin{split}
\mbox{TV}_{\textrm{sol}}(\bx(m))& =\parallel \mathbf{G}_{2}^{-1/2} \mB_2^{T} \bx(m) \parallel^{2}\\ &=\bx_{\textrm{sol}}(m)^{T} \mB_2 \mathbf{G}_{2}^{-1} \mB_2^{T} \bx_{\textrm{sol}}(m)
\end{split}
\eeq
where in  the last equality we exploited the orthogonality among the irrotational, harmonic and solenoidal subspaces. Note that since the  metric tensor is a diagonal matrix, equation (\ref{eq:TV_m})
can be expressed in the form:
\beq \label{eq:TV_m1}
\mbox{TV}_{\textrm{sol}}(\bx(m))= \sum_{i=1}^{T} w_2(i) \bx_{\textrm{sol}}(m)^{T} \mb_i\mb_{i}^{T} \bx_{\textrm{sol}}(m)
\eeq
where $w_2(i)$ is the $i$-th positive diagonal entry of the metric tensor  $\mathbf{G}_{2}^{-1}$, $T$ is the number of (filled) triangles of the complex and  $\mb_i$ is the $i$-th column of $\mB_2$.
Then, our goal is to find the optimal weights'  vector $\bw_2=[w_2(1), \ldots,w_2(T)]^{T}$  minimizing the total variation of the observed solenoidal signals. Therefore, denoting with $\bx_{\textrm{sol}}=[\bx_{\textrm{sol}}(1),\ldots, \bx_{\textrm{sol}}(M)]$ the $E\times M$ matrix with columns the observed signals, the metric learning problem can be formulated as 
\beq \label{eq:TV_min}
\begin{array}{lll}
 \underset{\mathbf{w}_2\in \mathbb{R}^T}{\text{min}} &  \ds \sum_{i=1}^{T} w_2^{2}(i) \text{tr}(\bX^{T}_{s} \mb_i \mb_i^{T} \bX_{s})\qquad \qquad \qquad (\mathcal{P})\\
 \; \; \text{s.t.} & \text{a)}\ds \sum_{i=1}^{T} w_2(i)=1,\\
 \; \; \quad &  \text{b)} \;  w_2(i) > 0, \quad \forall \, i 
 \end{array}
\eeq
where we  consider a quadratic objective function instead of the linear one leading in our problem  to a trivial solution. The constraint $\text{a)}$ forces  the sum of the positive weights to be a constant value while the constraints in $\text{b)}$ ensure positive variables $w_2(i)$.
To simplify our notation, let us introduce the positive coefficients $a_i=\text{tr}(\bX^{1 \, T}_{s} \mb_i \mb_i^{T} \bx_{\textrm{sol}})$. 
The optimization problem $\mathcal{P}$ admits a closed form solution as stated in the following theorem.
\begin{theorem}
\textit{Given any set of positive coefficients $\{ a_i\}_{i=1}^{T}$, the convex optimization problem $\mathcal{P}$ admits the closed form solution: \beq w_{2}(i)=\frac{\lambda^{\star}}{2 a_i} \eeq with $\lambda^{\star}=1/(\sum_{i=1}^{T}1/(2a_i))$.}
\end{theorem}
\begin{proof}
First let us observe that the objective function in problem $\mathcal{P}$ is a convex function being a  linear combination with positive coefficients of quadratic variables. Then, problem $\mathcal{P}$ is a convex optimization problem  since the linear constraint $\text{a)}$ and the  constraints in $\text{b)}$ define a feasible convex set.
Therefore, any optimal solution   $\bw^{\star}_2$  satisfies the KKT conditions of   $\mathcal{P}$ that are necessary and sufficient conditions for optimality (note that Slater's constraint qualification is satisfied). Then, denoting with $\mathcal{L}(\bw,\lambda,\boldsymbol{\mu})$ the Lagrangian function of $\mathcal{P}$, the KKT conditions are
\beq
\begin{array}{ll}
  \text{(i)} \; \ds \frac{\partial \mathcal{L}}{\partial w_2(i)}= 2 w_2(i) a_i-\lambda -\mu_i=0, \quad   \forall \, i & \medskip\\
  \text{(ii)} \;  \mu_i w_2(i)=0,  \;\; \mu_i\geq 0,  w_2(i)>0, \;\; \forall \; i & \medskip\\
   \text{(iii)}\;  \lambda(\sum_i w_2(i)-1 )=0, \;\;  \lambda \in \mathbb{R}, \;\; \sum_{i=1}^{T} w_2(i)=1. &
\end{array}
\eeq
Since $w_2(i)>0$, from $\text{(ii)}$ we get $\mu_i=0$, so that condition $\text{(i)}$ becomes $w_2(i)=\frac{\lambda}{2 a_i}$. Replacing these variables in the linear constraint, one gets $\lambda^{\star}=1/(\sum_i 1/(2 a_i))$ so that the optimal solutions are $w_2(i)=\frac{\lambda^{\star}}{2 a_i}$. 
\end{proof}\\
To numerically test the effectiveness of the proposed metric learning strategy we solve problem $\mathcal{P}$ to find the metric tensor from the observation of signals  over the edges of the simplicial complex. Specifically, we generated $100$   random geometric graphs composed of $N=40$ vertices by filling all the possible $2$-simplices in the graph. For each graph we  generated a metric tensor with diagonal random entries between $[0,1]$ such that $\bw$ belongs to the feasible set of $\mathcal{P}$. In Figure 
\ref{fig:metr},  we report the mean squared metric estimation error $\parallel {\bw}_2-\hat{\bw}_2 \parallel_F$ versus the number $M$ of observed edge signals. The results are averaged over $100$ simplicial complex realizations and, for each complex, by generating $100$ random matrix  $\bX_s$ of bandlimited  signals. We can observe that as the number $M$ of signals increases a more accurate estimation $\hat{\bw}_2$ of the true metric is provided.     \begin{figure}[t!]
\centering
\includegraphics[width=0.48\textwidth]{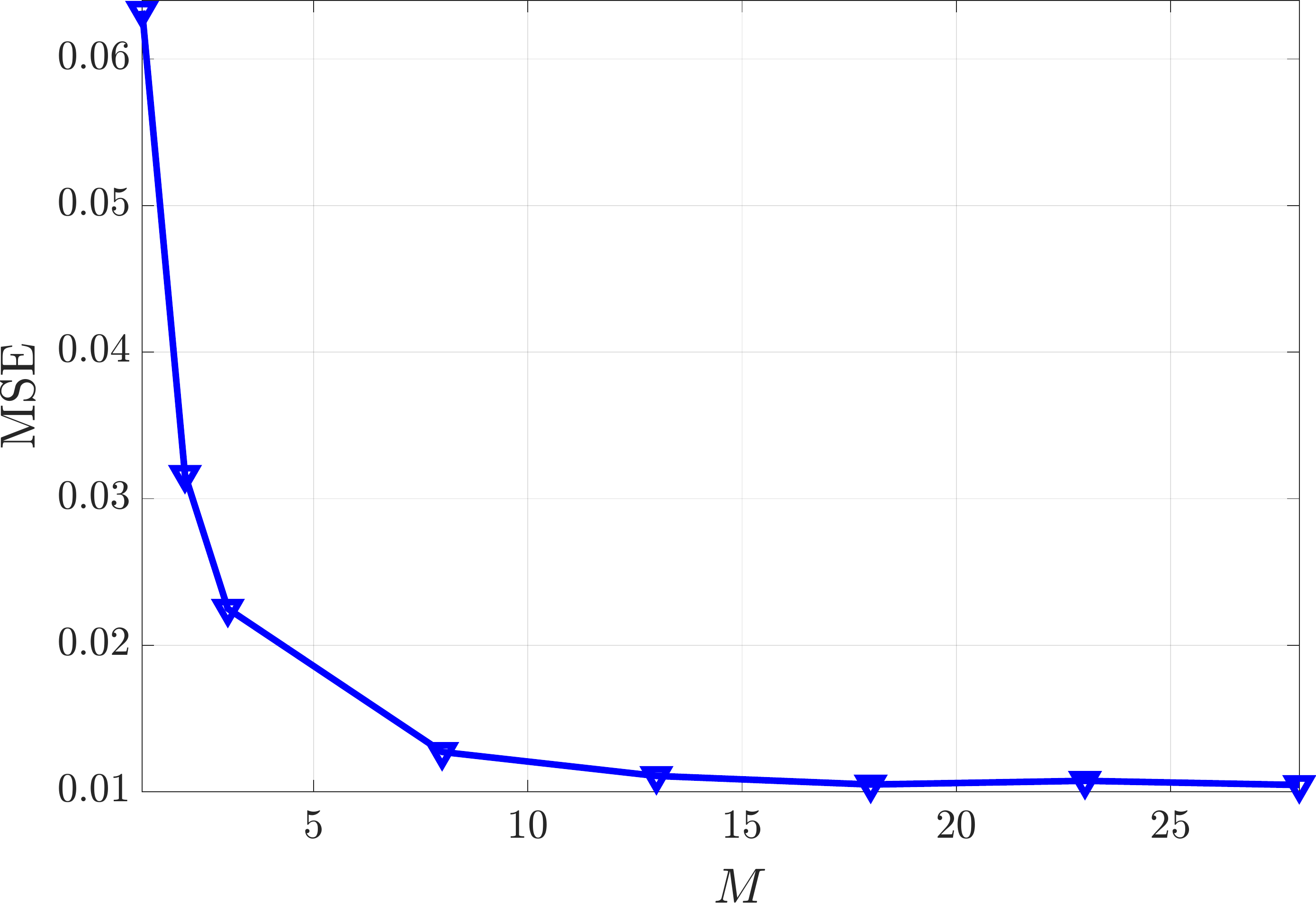}
\caption{ MSE versus the number of observed signals $M$.} \label{fig:metr}
\end{figure} 
\begin{section}{Comments and Conclusions}
     We presented topological signal processing tools for elaborating signals defined over weighted simplicial complexes. Specifically, relying on the weighted Hodge Laplacian theory, we proposed efficient strategies to jointly learn the weights of the complex and the filters for the solenoidal, irrotational and harmonic components of the signals defined over the complex. We numerically assessed the effectiveness of the proposed procedures. This is a preliminary work with two main goals: the first one is casting the algebraic topology notions of abstract simplicial complexes with non trivial metric tensors in a signal processing framework; the second one is proposing signal processing tools able to handle simplex-structured data defined on weighted simplicial complexes. We plan to extend the proposed framework both from a theoretical and an applied points of view.
\end{section}




\vfill\pagebreak
\newpage

\bibliographystyle{IEEEbib}

\bibliography{reference}

\begin{thebibliography}{10}

\bibitem{carlsson2009topology}
G.~Carlsson,
\newblock ``Topology and data,''
\newblock {\em Bulletin Amer. Math. Soc.}, vol. 46, no. 2, pp. 255--308, 2009.

\bibitem{munkres2000topology}
J.~R. Munkres,
\newblock {\em Topology},
\newblock Prentice Hall, 2000.

\bibitem{shuman2013}
D.~I. Shuman, S.~K. Narang, P.~Frossard, A.~Ortega, and P.~Vandergheynst,
\newblock ``The emerging field of signal processing on graphs: Extending
  high-dimensional data analysis to networks and other irregular domains,''
\newblock {\em IEEE Signal Process. Mag.}, vol. 30, no. 3, pp. 83--98, May
  2013.

\bibitem{ortega2018graph}
A.~Ortega, P.~Frossard, J.~Kova{\v{c}}evi{\'c}, J.~M.~F. Moura, and
  P.~Vandergheynst,
\newblock ``Graph signal processing: Overview, challenges, and applications,''
\newblock {\em Proc. of the IEEE}, vol. 106, no. 5, pp. 808--828, May 2018.

\bibitem{barb_2020}
S.~Barbarossa and S.~Sardellitti,
\newblock ``Topological signal processing over simplicial complexes,''
\newblock {\em IEEE Trans. Signal Process.}, vol. 68, pp. 2992--3007, Mar.
  2020.

\bibitem{barb_Mag_2020}
S.~Barbarossa and S.~Sardellitti,
\newblock ``Topological signal processing: Making sense of data building on
  multiway relations,''
\newblock {\em IEEE Signal Process. Mag.}, vol. 37, no. 6, pp. 174--183, Nov.
  2020.

\bibitem{SCHAUB2021}
M.~T. Schaub, Y.~Zhu, J.-B. Seby, T.~M. Roddenberry, and S.~Segarra,
\newblock ``Signal processing on higher-order networks: Livin’ on the edge...
  and beyond,''
\newblock {\em Signal Process.}, p. 108149, 2021.

\bibitem{jiang2011}
X.~Jiang, L.~H. Lim, Y.~Yao, and Y.~Ye,
\newblock ``Statistical ranking and combinatorial {H}odge theory,''
\newblock {\em Math. Program.}, vol. 127, no. 1, pp. 203--244, 2011.

\bibitem{roman2015simplicial}
T.~Roman, A.~Nayyeri, B.~T. Fasy, and R.~Schwartz,
\newblock ``A simplicial complex-based approach to unmixing tumor progression
  data,''
\newblock {\em BMC Bioinfor.}, vol. 16, no. 1, pp. 254, 2015.

\bibitem{andjelkovic2020topology}
M.~Andjelkovi{\'c}, B.~Tadi{\'c}, and R.~Melnik,
\newblock ``The topology of higher-order complexes associated with brain hubs
  in human connectomes,''
\newblock {\em Sci. Rep.}, vol. 10, no. 1, pp. 1--10, 2020.

\bibitem{estrada2018centralities}
E.~Estrada and G.~J. Ross,
\newblock ``Centralities in simplicial complexes. {A}pplications to protein
  interaction networks,''
\newblock {\em Jour. Theoret. Biology}, vol. 438, pp. 46--60, 2018.

\bibitem{san}
L.~Giusti, C.~Battiloro, P.~Di~Lorenzo, S.~Sardellitti, and S.~Barbarossa,
\newblock ``Simplicial attention neural networks,''
\newblock {\em ArXiv}, vol. abs/2203.07485, 2022.

\bibitem{bodnar2021weisfeiler}
C.~Bodnar, F.~Frasca, N.~Otter, Y.~G. Wang, P.~Lio, G.~Montufar, and
  M.~Bronstein,
\newblock ``Weisfeiler and {L}ehman go cellular: {CW} networks,''
\newblock {\em Adv. in Neural Inform. Process. Systems}, vol. 34, pp.
  2625--2640, 2021.

\bibitem{scn}
M.~Yang, E.~Isufi, and G.~Leus,
\newblock ``Simplicial convolutional neural networks,''
\newblock in {\em 2022 IEEE Int. Conf. on Acous., Speech and Signal Process.
  (ICASSP)}, 2022, pp. 8847--8851.

\bibitem{bianconi2022weighted}
F.~Baccini, F.~Geraci, and G.~Bianconi,
\newblock ``Weighted simplicial complexes and their representation power of
  higher-order network data and topology,''
\newblock {\em Phys. Rev. E}, vol. 106, no. 3, pp. 034319, 2022.

\bibitem{bianconi2017weighted}
O.~T. Courtney and G.~Bianconi,
\newblock ``Weighted growing simplicial complexes,''
\newblock {\em Physical Review E}, vol. 95, no. 6, pp. 062301, 2017.

\bibitem{sharma2017weighted}
A.~Sharma, T.~J. Moore, A.~Swami, and J.~Srivastava,
\newblock ``Weighted simplicial complex: A novel approach for predicting small
  group evolution,''
\newblock in {\em Pacific-Asia Conf. on Knowl. Discov. and Data Mining}.
  Springer, 2017, pp. 511--523.

\bibitem{Wu2018}
C.~Wu, S.~Ren, J.~Wu, and K.~Xia,
\newblock ``Weighted (co)homology and weighted {L}aplacian,''
\newblock {\em arXiv preprint arXiv:1804.06990}, 2018.

\bibitem{meng2020weighted}
Z.~Meng, D.~V. Anand, Y.~Lu, J.~Wu, and K.~Xia,
\newblock ``Weighted persistent homology for biomolecular data analysis,''
\newblock {\em Scient. rep.}, vol. 10, no. 1, pp. 1--15, 2020.

\bibitem{hajiji2022att}
M.~{Hajij}, G.~{Zamzmi}, T.~{Papamarkou}, N.~{Miolane},
  A.~{Guzm{\'a}n-S{\'a}enz}, and K.~{Natesan Ramamurthy},
\newblock ``{Higher-Order Attention Networks},''
\newblock {\em arXiv e-prints}, p. arXiv:2206.00606, June 2022.

\bibitem{Barbarossa2020TopologicalSP}
Sergio Barbarossa and Stefania Sardellitti,
\newblock ``Topological signal processing over simplicial complexes,''
\newblock {\em IEEE Transactions on Signal Processing}, vol. 68, pp.
  2992--3007, 2020.

\bibitem{grady2010}
L.~J. Grady and J.~R. Polimeni,
\newblock {\em Discrete calculus: Applied analysis on graphs for computational
  science},
\newblock Sprin. Scie. \& Busin. Media, 2010.

\bibitem{frankel2011geometry}
T.~Frankel,
\newblock {\em The geometry of physics: an introduction},
\newblock Cambridge Univ. Press, 2011.

\bibitem{desbrun2006discrete}
M.~Desbrun, E.~Kanso, and Y.~Tong,
\newblock ``Discrete differential forms for computational modeling,''
\newblock in {\em Discrete Differential Geometry}, A.~I. Bobenko, J.~M.
  Sullivan, P.~Schr{\"o}der, and G.~M. Ziegler, Eds., pp. 287--324.
  Birkh{\"a}user Basel, 2008.

\bibitem{Bell2008}
N.~Bell and L.~N. Olson,
\newblock ``Algebraic multigrid for k-form laplacians,''
\newblock {\em Numer. Linear Algebra with Applic.}, vol. 15, no. 2-3, pp.
  165--185, 2008.

\end{thebibliography}
\end{document}